\documentclass[11pt,twoside]{article}

\usepackage{mathtools} 
\usepackage{mathrsfs} 
\usepackage[symbol]{footmisc}

\newcommand{\hn}{\mathfrak{h}_{\mathrm{ph}}}

\newcommand{\FF}{\mathcal{F}}
\newcommand{\HH}{\mathscr H}

\newcommand{\HHel}{\mathscr H_{\mathrm{el}}}
\newcommand{\HHf}{\mathscr H_{\mathrm{f}}}

\newcommand{\Hf}{H_{\mathrm{f}}}
\newcommand{\Helf}{H_{\mathrm{el} + \mathrm{f}}}

\newcommand{\iu}{\mathrm{i}}   
\newcommand{\dd}{\mathrm{d}}   

\usepackage[utf8]{inputenc} 
\usepackage{amsmath,amssymb,amsthm}
\usepackage{pstricks,pst-node,pst-coil,pst-plot,pstricks-add}
\usepackage{geometry,epsfig}
\usepackage{bbm}
\usepackage{cite}

\usepackage{cite}
\usepackage{paralist}
\usepackage{cleveref}
\usepackage{enumitem}
\usepackage{emptypage}

\DeclareGraphicsExtensions{.pdf}

\usepackage{pgfplots}
\pgfplotsset{compat=1.18}
\usepackage{amsmath}

\newtheorem{lemma}{Lemma}
\newtheorem{theorem}[lemma]{Theorem}

\newcommand{\C}{\mathbb{C}} 
\newcommand{\R}{\mathbb{R}} 

\newcommand{\eps}{\varepsilon}

\title{\textbf{Self-Adjointness of the Standard Model of Non-Relativistic QED}}

\author{V.~Ku{\ss}maul\footnote{valentin.kussmaul@mathematik.uni-stuttgart.de}\\  
\small Fachbereich Mathematik, Universit\"at Stuttgart, D-70569 Stuttgart, Germany}  

\begin{document}
\maketitle

\begin{abstract}
For systems of non-relativistic charged particles minimally coupled to the soft modes of the quantized radiation field, we prove self-adjointness on the domain of the free Hamiltonian. This result is not new, but the proof we 
give shortens previous arguments. It is based on a comparison of graph norms and Nelson's commutator theorem. 
\end{abstract}

\section{Introduction} \label{introduction}
The \textit{Standard Model of Non-Relativistic Quantum Electrodynamics} \cite{BFS98}, also known as Pauli-Fierz Hamiltonian \cite{Hiroshima02,HH08}, describes the interaction of charged particles (e.g., electrons and protons) with the quantized electromagnetic field. The model is considered the standard theory in non-relativistic energy regimes. To ensure mathematical well-definedness, the particles are coupled to the low-energy modes of the field only. Physically, such a cutoff is legitimate for a theory on low-energy phenomena. 

The starting point for a rigorous analysis of any quantum system is the proof of (essential) self-adjointness. It is well known that the Pauli-Fierz Hamiltonian is self-adjoint 
on the domain of the free Hamiltonian and essentially self-adjoint on every core of the free Hamiltonian \cite{BFS99,Hiroshima02,HH08}. For sufficiently small values of the coupling constant, depending on the UV-cutoff, this follows from the Kato-Rellich theorem \cite{BFS99}. Self-adjointness for arbitrary coupling strengths and arbitrary UV-cutoff has so far been proven by implicitly defining a self-adjoint operator via a suitable mathematical object and then showing that the operator obtained this way agrees with the Pauli-Fierz Hamiltonian \cite{Hiroshima02,HH08}. The aforementioned mathematical object is a functional integral formula for the semigroup in \cite{Hiroshima02} and a closed quadratic form in \cite{HH08}.

The purpose of this paper is to provide a concise and direct proof of the results in \cite{Hiroshima02,HH08}. The strategy is as follows: For a symmetric operator, self-adjointness is equivalent to closedness and essential self-adjointness. The former follows from graph norm equivalence to the free operator, and the latter from Nelson's commutator theorem. The advantage is that the Pauli-Fierz Hamiltonian is analyzed directly, avoiding a detour via a second implicitly defined operator. 

This paper is organized as follows. In \Cref{sec:model-and-result}, we introduce the Pauli-Fierz Hamiltonian and state the main result, \Cref{main-result}. In  \Cref{sec:proof}, we prove the main result. The appendix collects technical estimates.

\section{Mathematical model and main result} \label{sec:model-and-result}
In this section, we briefly introduce the model and state the main result. More elaborate introductions can be found in \cite{BFS98,HH08}. For simplicity, we consider $N$ identical spin-$1/2$ particles (electrons) coupled to the radiation field, in the presence of static nuclei. We could also allow for a collection of particles with various charges.

The Hilbert space for the electrons is the anti-symmetric tensor product
\begin{align} \label{HHel}
	\HHel = \otimes_{\mathrm{asym}}^N (L^2(\R^3) \otimes \C^2).
\end{align}
The repulsion between electrons and the attraction exerted by atomic nuclei is described by the many-body potential 
\begin{align*}
	V(\mathbf{x}_1, \dots, \mathbf{x}_N) = \sum_{j = 1}^N v(\mathbf{x}_j) + \sum_{j < k} w(\mathbf{x}_j - \mathbf{x}_k)
\end{align*}
where $v, w : \R^3 \to \R$ are (sums of) Coulomb potentials or, more generally, functions in $L^2(\R^3) + L^\infty(\R^3)$ with $w(\mathbf{x}) = w(-\mathbf{x})$ for a.e. $\mathbf{x} \in \R^3$. 

The Hilbert space for the radiation field is the symmetric Fock space
\begin{align*}
	\HHf = \FF = \bigoplus_{n = 0}^\infty \otimes^n_\mathrm{sym} \hn
\end{align*}
where $\hn = L^2(\R^3 \times \{ 1, 2 \})$ describes a single photon with momentum $\mathbf{k} \in \R^3$ and polarization $\lambda \in \{1, 2\}$. In $\HHf$ let $a_\lambda(\mathbf{k})$ and $a_\lambda^*(\mathbf{k})$ denote the usual annihilation and creation operators, formally satisfying the canonical commutation relations (CCR), 
\begin{align*}
	[a_{\lambda_1}^{\#}(\mathbf{k}_1), a_{\lambda_2}^{\#}(\mathbf{k}_2)] = 0, \quad [a_{\lambda_1}(\mathbf{k}_1), a_{\lambda_2}^{*}(\mathbf{k}_2)] = \delta_{\lambda_1, \lambda_2} \, \delta(\mathbf{k}_1 - \mathbf{k}_2),
\end{align*}
where $a^{\#} = a$ or $a^*$. Let $|\mathbf{k}|$ be the energy of a photon with momentum $\mathbf{k}$ and 
\begin{align*}
	\Hf = \sum_{\lambda = 1,2}  \int \hspace{-1mm}\dd^3k \, |\mathbf{k}| a_\lambda^*(\mathbf{k}) a_\lambda(\mathbf{k}) 
\end{align*}
the field energy operator. The $j$th electron couples to the radiation field via the quantized vector potential
\begin{align*} 
	\mathbf{A}(\mathbf{x}_j) = \sum_{\lambda = 1,2} \int \hspace{-1mm} \frac{\dd^3k}{\sqrt{|\mathbf{k}|}}  \left[ e^{-\iu \mathbf{k}\cdot \mathbf{x}_j} \rho(\mathbf{k}) \boldsymbol{\eps}_\lambda(\mathbf{k}) a^*_\lambda(\mathbf{k})  + e^{\iu \mathbf{k}\cdot \mathbf{x}_j} \overline{\rho}(\mathbf{k})  \overline{\boldsymbol{\eps}}_\lambda(\mathbf{k}) a_\lambda(\mathbf{k}) \right] \quad (\mathbf{x}_j \in \R^3).
\end{align*}
Here, for a.e. $(\mathbf{k}, \lambda) \in \R^3 \times \{1, 2 \}$, the polarization vectors $\boldsymbol{\eps}_\lambda(\mathbf{k}) \in \C^3$ are measurable and chosen such that $\mathbf{k}/|\mathbf{k}|, \boldsymbol{\eps}_1(\mathbf{k}), \boldsymbol{\eps}_2(\mathbf{k})$ are orthonormal. The function $\rho: \R^3 \to \C$ is assumed to be measurable with
\begin{align} \label{UV-condition}
	\int \left(|\mathbf{k}| + \frac{1}{|\mathbf{k}|^2}\right) |\rho(\mathbf{k})|^2 \dd^3k < \infty.
\end{align}
It absorbs the coupling constant and implements a UV-cutoff.

In $\HH_{\mathrm{el} + \mathrm{f}} = \HHel \otimes \HHf$ the Hamiltonian of the total system is defined on the domain of the free Hamiltonian $H_0 = -\Delta + \Hf$ as the sum of operator compositions
\begin{align} \label{standard-model}
	\Helf = \sum_{j = 1}^N \left[(-\iu \boldsymbol{\nabla}_{\mathbf{x}_j} + \mathbf{A}(\mathbf{x}_j))^2 + \boldsymbol{\sigma}_j  \cdot \mathbf{B}(\mathbf{x}_j) \right] + V + \Hf
\end{align}
where $\mathbf{B}(\mathbf{x}_j) = \mathrm{rot} \, \mathbf{A}(\mathbf{x}_j)$ and $\boldsymbol{\sigma}_j$ denotes the triple of Pauli matrices acting on the $j$th factor of \eqref{HHel}. Condition \eqref{UV-condition} guarantees that $\Helf$ is well-defined on the domain of $H_0$. 
\begin{theorem}[\hspace{-0.1mm}\cite{Hiroshima02,HH08}] \label{main-result}
The Hamiltonian $H_{\mathrm{el} + \mathrm{f}}$ is self-adjoint and bounded from below. The graph norms of $H_{\mathrm{el} + \mathrm{f}}$ and $H_0$ are equivalent and $\Helf$ is essentially self-adjoint on any core of $H_0$.
\end{theorem}

\noindent
\emph{Remarks.} 
\begin{enumerate}
\item
It suffices to prove the theorem neglecting the terms $\boldsymbol{\sigma}_j  \cdot \mathbf{B}(\mathbf{x}_j)$ and $V$, since these can be added later using the Kato-Rellich theorem. Self-adjointness of $\Helf$ is then shown by establishing essential self-adjointness and closedness. Essential self-adjointness follows from Nelson's commutator theorem, which is applicable only if magnetic field and potential are neglected \cite{KMS13}. Closedness follows from graph norm equivalence of $\Helf$ and $H_0$, which is shown using estimates from \cite{HH08}. 
\item
In \cite{HH08} a self-adjoint realization of \eqref{standard-model}, called $T_\mathbf{A}$, is defined via a semibounded and closed quadratic form. Self-adjointness of $\Helf$ is shown by proving $\Helf = T_\mathbf{A}$. The easy part is $\Helf \subset T_\mathbf{A}$ \cite[Lemma 11 (b - c)]{HH08}. Since $D(\Helf) = D(H_0)$, this entails the inclusion $D(H_0)\subset D(T_\mathbf{A})$, which, by the closed graph theorem, implies that $T_\mathbf{A}$ is $H_0$-bounded. This relative bound, combined with an approximation argument that shows that $D(H_0)$ is a core for $T_\mathbf{A}$ \cite[Lemma 11 (d)]{HH08}, implies that any core for $H_0$ is a core for $T_\mathbf{A}$. Therefore, to prove $D(H_0) \supset D(T_\mathbf{A})$ (and hence $\Helf = T_\mathbf{A}$), it suffices to show that $H_0$ is $T_\mathbf{A}$-bounded on a core for $H_0$, which is done by explicit computation in Lemma 12 of \cite{HH08}. In \cite{Hiroshima02} a representation formula for the semigroup of \eqref{standard-model} is used to define a self-adjoint operator equal to $\Helf$.
\item 
In an upcoming publication we prove \emph{essential} self-adjointness of $\Helf$ in more general situations where the potential $V$ is singular and possibly not relatively bounded with respect to $\Delta$ \cite{GK26}. 
\end{enumerate}

\section{Proof of \Cref{main-result}} \label{sec:proof}
Let $\mathbf{p} = (-\iu \boldsymbol{\nabla}_{{\mathbf{x}}_1}, \dots, -\iu\boldsymbol{\nabla}_{{\mathbf{x}}_N})$ and $\mathbf{A} = (\mathbf{A}(\mathbf{x}_1), \dots, \mathbf{A}(\mathbf{x}_N))$. 
With a similar notation for magnetic field and Pauli matrices, the Hamiltonian \eqref{standard-model} simply reads
\begin{align*}
	H_{\mathrm{el} + \mathrm{f}} = (\mathbf{p} + \mathbf{A})^2 + \boldsymbol{\sigma} \cdot \mathbf{B} + V + \Hf.
\end{align*}
We first analyze
\begin{align*} 
	H = (\mathbf{p} + \mathbf{A})^2 + \Hf.
\end{align*}
Let $H_0 = \mathbf{p}^2 + \Hf$ be the free Hamiltonian. On the domain of $H_0$ we have that
\begin{align}
	H &= \mathbf{p}^2 + \mathbf{p} \cdot \mathbf{A}+ \mathbf{A} \cdot \mathbf{p} + \mathbf{A}^2 + \Hf \nonumber \\ 
	&=  H_0+ 2 \mathbf{A} \cdot \mathbf{p} + \mathbf{A}^2 \label{decomposed}
\end{align}
where we used that $\mathbf{k} \perp \boldsymbol{\eps}_\lambda(\mathbf{k})$ implies Coulomb gauge $(\mathbf{p} \cdot \mathbf{A}) = 0$. From $\eqref{decomposed}$ $H$ is easily seen to be $H_0$-bounded. Indeed, by \Cref{standard-bounds}, $\mathbf{A}^2$ is $\Hf$- hence $H_0$-bounded. $\mathbf{A}$ is $(\Hf + 1)^{1/2}$-bounded, so $\mathbf{A} \cdot \mathbf{p}$ is $(\Hf + 1)^{1/2} \mathbf{p}$- hence $H_0$-bounded. The following lemma proves the reverse bound:
\begin{lemma}[\hspace{-0.1mm}\cite{HH08}] \label{graph-norm-equivalence}
On the domain of $H_0$ the graph norms of $H$ and $H_0$ are equivalent. 
\end{lemma}
\begin{proof}
We follow the argument in the proof of Lemma 12 in \cite{HH08}. Since $H$ is $H_0$-bounded, it suffices to prove the bound $\| H_0 \psi \| \leq C (\| H \psi \| + \| \psi \|)$ on a core of $H_0$ such as the algebraic tensor product of test functions with finite particle vectors in the domain of $\Hf$. We proceed in two steps: \newline 
\medskip\noindent
\underline{\emph{Step 1}}.  $\| H_0 \psi \| \leq C (\| (\mathbf{p} + \mathbf{A})^2 \psi \| + \| \Hf \psi \| + \| \psi \|)$.

Since $\mathbf{p}^2 = (\mathbf{p} + \mathbf{A})^2 - 2 \mathbf{A} \cdot (\mathbf{p} + \mathbf{A}) + \mathbf{A}^2$ and $\mathbf{A}^2$ is $\Hf$-bounded, we need to show that 
\begin{align*}
	\| \mathbf{A} \cdot (\mathbf{p} + \mathbf{A}) \psi \| \leq C (\| (\mathbf{p} + \mathbf{A})^2 \psi \| + \| \Hf \psi \| + \| \psi \|).
\end{align*}
Indeed, $\| \mathbf{A} \cdot (\mathbf{p} + \mathbf{A}) \psi \| \leq C \| (\Hf + 1)^{1/2} (\mathbf{p} + \mathbf{A}) \psi \|$ and 
\begin{align*}
	 (\mathbf{p} + \mathbf{A}) (\Hf + 1) (\mathbf{p} + \mathbf{A}) &= (\mathbf{p} + \mathbf{A}) [\Hf, \mathbf{A}] + (\mathbf{p} + \mathbf{A})^2 (\Hf + 1) \\ 
	 &\leq C([(\mathbf{p} + \mathbf{A})^2]^2 + \Hf^2 + 1)
\end{align*}
where we used that $ [\Hf, \mathbf{A}]$ is $\Hf^{1/2}$-bounded by \Cref{standard-bounds}. \newline
\medskip\noindent
\underline{\emph{Step 2}}. $\| (\mathbf{p} + \mathbf{A})^2 \psi \| + \| \Hf \psi \| \leq C (\| H \psi \| + \| \psi \|)$. 

In form sense, we have that 
\begin{align} \label{algebra}
	H^2 &= [(\mathbf{p} + \mathbf{A})^2]^2 + \Hf^2 + (\mathbf{p} + \mathbf{A})^2 \Hf + \Hf (\mathbf{p} + \mathbf{A})^2 \nonumber \\ 
	&= [(\mathbf{p} + \mathbf{A})^2]^2 + \Hf^2 + 2 (\mathbf{p} + \mathbf{A}) \Hf (\mathbf{p} + \mathbf{A}) + [\mathbf{p} + \mathbf{A}, [\mathbf{p} + \mathbf{A}, \Hf]].
\end{align}
From $[\mathbf{p}, \mathbf{A}] = 0$, $[\mathbf{p}, \Hf] = \boldsymbol{0}$ and the CCR, the double commutator evaluates to 
\begin{align} \label{matrix-element}
	[\mathbf{p} + \mathbf{A}, [\mathbf{p} + \mathbf{A}, \Hf]] = -4 N \int | \rho(\mathbf{k}) |^2 \dd^3k > -\infty.
\end{align}
Step 2 follows from \eqref{algebra}, \eqref{matrix-element} and $\Hf \geq 0$.
\end{proof}

\begin{lemma}[\hspace{-0.1mm}\cite{KMS13}] \label{ess-sa}
The operator $H$ is essentially self-adjoint on any core of $H_0$. 
\end{lemma}
\begin{proof}
We follow the argument in the proof of Lemma 3.1 in \cite{KMS13}.
We introduce the auxiliary operator $N_\lambda = (\mathbf{p} + \mathbf{A})^2 + \lambda \Hf + 1$. By Kato-Rellich, if $\lambda \gg 1$, $N_\lambda$ is self-adjoint on the domain of $N_{0, \lambda} = \mathbf{p}^2 + \lambda \Hf + 1$. Indeed, by the arguments following \eqref{decomposed}, $2 \mathbf{A} \cdot \mathbf{p} + \mathbf{A}^2$ is $N_{0, \lambda}$-bounded with relative bound $O(\lambda^{-1/2} + \lambda^{-1})$. 
We apply Nelson's commutator theorem\footnote{See Theorem X.37 in \cite{RS75}.}: First, in form sense,
\begin{align*}
	\pm \iu [H, N_\lambda] &= \pm (\lambda - 1) \big( (\mathbf{p} + \mathbf{A}) \cdot \iu [\mathbf{A}, \Hf]  + \iu [\mathbf{A}, \Hf] \cdot (\mathbf{p} + \mathbf{A})  \big) \\ 
	& \leq C_\lambda N_\lambda
\end{align*}
where we used that $[\mathbf{A}, \Hf]$ is $\Hf^{1/2}$-bounded by \Cref{standard-bounds}. Second, $H$ is $N_\lambda$-bounded because $H$ is $H_0$-bounded and the graph norms of $H_0$, $N_{0, \lambda}$ and $N_\lambda$ are equivalent. 
Hence $H$ is essentially self-adjoint on any core of $N_{\lambda}$ (or $H_0$).
\end{proof}

\begin{proof}[Proof of \Cref{main-result}]
From the graph norm equivalence of $H$ and $H_0$, see \Cref{graph-norm-equivalence}, it follows that $H$ is a closed operator on the domain of $H_0$. By \Cref{ess-sa}, $H$ is also essentially self-adjoint, so $H$ is self-adjoint. 
\Cref{main-result} now follows from the Kato-Rellich theorem since $\boldsymbol{\sigma} \cdot \mathbf{B} + V$ is infinitesimally $H_0$- hence $H$-bounded. 
(More precisely, $V$ is infinitesimally $\Delta$-bounded and, by \Cref{standard-bounds}, $\boldsymbol{\sigma} \cdot \mathbf{B} $ is $\Hf^{1/2}$-bounded hence infinitesimally $\Hf$-bounded.)
\end{proof}

\noindent
\emph{Acknowledgement.}
This work is supported by the German Research Foundation (DFG) under Grant GR 3213/4-1. The author thanks Marcel Griesemer for helpful comments on the manuscript. 

\appendix 

\setcounter{lemma}{0}
\renewcommand{\thelemma}{\Alph{lemma}}

\section{Appendix}\label{sec:appendix}

\begin{lemma} \label{standard-bounds}
With the notation from \Cref{sec:proof}, we have that 
\begin{enumerate}
\item[(i)] $\mathbf{A}$ is $\Hf^{1/2}$-bounded and $\mathbf{A}^2$ is $\Hf$-bounded.
\item[(ii)] $[\Hf, \mathbf{A}]$ and $\boldsymbol{\sigma} \cdot \mathbf{B}$ are $\Hf^{1/2}$-bounded.
\end{enumerate}
\end{lemma}
\begin{proof}
For $G \in L^2(\R^3 \times \{1, 2 \})$ let $a(G) = \sum_{\lambda } \int \hspace{-1mm}\dd^3k \,  \overline{G(\mathbf{k}, \lambda)} a_\lambda(\mathbf{k})$ in $\HHf$.
If $G(\mathbf{k}, \lambda)/\sqrt{|\mathbf{k}|} \in L^2(\R^3 \times \{ 1, 2 \})$, then, by Hölder's inequality, 
 \begin{align}
 	\| a(G) \psi \|&\leq \left( \sum_\lambda \int \hspace{-1mm} \frac{|G(\mathbf{k}, \lambda)|^2}{|\mathbf{k}|} \dd^3k \right)^{1/2} \| \Hf^{1/2} \psi \|, \label{single-bound}\\ 
	\| a(G) a(G) \psi \|&\leq \left( \sum_\lambda \int \hspace{-1mm} \frac{|G(\mathbf{k}, \lambda)|^2}{|\mathbf{k}|} \dd^3k \right) \| \Hf \psi \|. \label{double-bound}
 \end{align}
 From \eqref{single-bound}, \eqref{double-bound}, and the CCR, (i) follows. (ii) follows from \eqref{single-bound} and the CCR.
\end{proof}


\end{document}